\documentclass[11pt]{article}[12pt]
\usepackage[utf8]{inputenc}
\usepackage[english]{babel}
\usepackage{amsmath}
\usepackage{amssymb}
\usepackage{amsthm}
\usepackage{graphicx}
\usepackage{epstopdf}
\usepackage[affil-it]{authblk}
\usepackage{upgreek}

\usepackage{hyperref}
\hypersetup{
    colorlinks=true,
    linkcolor=red,
    citecolor=blue,
    urlcolor=blue
    }

\usepackage[a4paper, total={6in, 9.6in}]{geometry}

\def\R{\mathbb{R}}

\newcommand{\pd}[2]{\dfrac{\partial #1}{\partial #2}}
\newcommand{\smallpd}[2]{{\partial #1}/{\partial #2}}
\newcommand{\deltapd}[2]{\dfrac{\delta #1}{\delta #2}}

\newcommand{\der}[2]{\dfrac{\mathrm{d} #1}{\mathrm{d} #2}}
\newcommand{\bra}[1]{\langle #1 |}
\newcommand{\ket}[1]{| #1 \rangle}

\newtheorem{theorem}{Theorem}
\newtheorem{proposition}{Proposition}
\newtheorem{remark}{Remark}

\def\F{\mathcal{F}}     
\def\J{\mathcal{J}}     
\def\L{\mathcal{L}}     
\def\M{\mathcal{M}}     
\def\r{\mathbf{r}}      
\def\I{\mathbb{I}}      
\def\d{\mathrm{d}}      

\title{\textbf{GRAPE optimization for open quantum systems with time-dependent decoherence rates driven by coherent and incoherent controls}}

\author[*]{Vadim N. Petruhanov}
\author[**]{Alexander N. Pechen}
\affil[$\,$]{\it Department of Mathematical Methods for Quantum Technologies,\par Steklov Mathematical Institute of Russian Academy of Sciences,\par Moscow, Russia;}
\affil[*]{vadim.petrukhanov@gmail.com}
\affil[**]{apechen@gmail.com}
\begin{document}

\maketitle

\begin{abstract}
The GRadient Ascent Pulse Engineering (GRAPE) method is widely used for optimization in quantum control. GRAPE is gradient search method based on exact expressions for gradient of the control objective. It has been applied to coherently controlled closed and open quantum systems. In this work, we adopt GRAPE method for optimizing objective functionals for open quantum systems driven by both coherent and incoherent controls. In our case, the tailored or engineered environment acts on the system as control via {\it time-dependent decoherence rates} $\gamma_k(t)$ or, equivalently, via {\it spectral density of the environment} $n_\omega(t)$. To develop GRAPE approach for this problem, we compute gradient of various objectives for general N-level open quantum systems both for piecewise class of control. The case of a single qubit is considered in details and solved analytically. For this case, an explicit analytical expression for evolution and objective gradient is obtained via diagonalization of a $3\times 3$ matrix determining the system's dynamics in the Bloch ball. The diagonalization is obtained by solving a cubic equation via Cardano's method.  The efficiency of the algorithm is demonstrated through numerical simulations for the state-to-state transition problem and its complexity is estimated.
\end{abstract}

\section{Introduction}\label{Sec0}
Control of individual quantum systems is an important tool necessary for development of modern quantum technologies~\cite{KochEPJQuantumTechnol2022,ShapiroBrumerBook2003,TannorBook2007,LetokhovBook2007,FradkovBook2007,BrifChakrabartiRabitz2010,WisemanMilburnBook2010,DongPetersen2010,Moore2011,CongBook2014,CPKoch2016OpenQS,DAlessandroBook2021}. In experimental circumstances, often controlled systems can not be isolated from the environment, so that they are open quantum systems. The presence of the environment is often considered as an obstacle for controlling the system. In other cases, the environment can be used as a useful resource for active controlling quantum systems via its temperature, pressure, or more generally, non-equilibrium spectral density, as for example, via incoherent control~\cite{Pechen_Rabitz_2006,PechenPRA2011} or in other approaches with using an engineered environment.  By incoherent control we mean control by (generally) time-dependent decoherence rates.

The approach for controlling open quantum systems via \textit{time-dependent decoherence rates} $\gamma_k(t)$ (or, equivalently, via \textit{time-dependent spectral density of the environment} $n(t)$) jointly with standard coherent control, with master equation 
\begin{equation}\label{eq:ME}
\frac{d\rho_t}{dt}=-i[H_0+H_c(t),\rho_t]+\sum\limits_k \gamma_k(t) {\cal D}_k,
\end{equation}
where ${\cal D}_k$ is a Gorini–Kossakowski–Sudarshan–Lindblad (GKSL) dissipator and $H_c(t)$ is coherent control Hamiltonian, was proposed in~\cite{Pechen_Rabitz_2006}, where in addition to general consideration two physical classes of the environment were exploited --- incoherent photons and quantum gas, with two explicit forms of ${\cal D}_k$ derived in the weak coupling and low density limits, respectively. In the present work, we derive gradient-based optimization scheme for finding optimal coherent controls and $\gamma_k(t)$ (expressed via $n(t)$) for various control objectives. In~\cite{PechenPRA2011}, it was shown that for the master equation with ${\cal D}_k$ derived in the weak coupling limit (describing atom interacting with photons) generic $N$-level quantum systems become approximately completely controllable in the set of density matrices. Exact description of reachable sets for a qubit via geometric control tehcnique was obtained in~\cite{Lokutsievskiy_2021}.  The~first experimental realization of Kraus maps studied in~\cite{Wu_2007_5681} for an open single qubit was done in~\cite{Zhang_Saripalli_Leamer_Glasser_Bondar_2022}.

Coherent control of closed systems induces  unitary dynamics and steers one pure state to another pure state, whereas incoherent control allows to produce mixed states from pure states. Incoherent control and engineered environment can be used, for example, for creation of mixed quantum states in some models of quantum computation with mixed states and non-unitary gates~\cite{Aharonov, Tarasov2002, Verstraete2009, Schulte-Herbruggen}, in nuclear magnetic resonance~\cite{LapertPRA2013}, steady quantum states generation \cite{science2015}, improving quantum state engineering and computation~\cite{Verstraete2009}, cooling of translational motion without reliance on internal degrees of freedom~\cite{Calarco2011}, preparing entangled states~\cite{Diehl_nature_2008,Weimer_nature_2010}, inducing multiparticle entanglement dynamics~\cite{Barreiro_nature_2010}, making robust quantum memories~\cite{Pastawski_2011}, creating entanglement dynamics in two-qubit systems~\cite{Ai_Fan_Jin_Cheng_2014}, and in other applications. 

Sometimes a solution for the optimal shape of the control can be obtained analytically. However, generally it is not the case and various numerical optimization methods are used, including GRadient Ascent Pulse Engineering (GRAPE) numerical optimization algorithm~\cite{khaneja_optimal_2005}, gradient flows~\cite{Glaser2010}, Krotov method~\cite{Tannor1992,Morzhin2019}, genetic algorithms for coherent control of closed systems~\cite{Judson1992} and incoherent control of open quantum systems in~\cite{Pechen_Rabitz_2006}, gradient free CRAB optimization algorithm ~\cite{CanevaPRA2011}, Hessian based optimization as in the Broyden–Fletcher–Goldfarb–Shanno (BFGS) algorithm and combined approaches ~\cite{EitanPRA2011,DalgaardPRA2020}, machine learning such as quantum reinforcement learning with incoherent control~\cite{Dong_Chen_Tarn_Pechen_Rabitz_2008}, deep reinforcement learning~\cite{Niu_Boixo_Smelyanskiy_Neven_2019}, autoencoders~\cite{Zauleck_VivieRiedle_2018}, speed gradient algorithm~\cite{PechenBorisenokFradkov2022}, Lyapunov control~\cite{Hou_Yi_2020} and various schemes~\cite{Palittapongarnpim_Sanders_NaturePhysics_2018,Torlai_Mazzola_Carrasquilla_et_al_NaturePhysics_2018,Gao_Qiao_Jiao_Ma_Hu_et_al_2018,Luchnikov_Vintskevich_Grigoriev_Filippov_2020,Morzhin_Pechen_LJM2021}. Monotonically convergent optimization in quantum control using Krotov's method was obtained for a large class of quantum control problems~\cite{Reich2012}. A general framework for sequential update algorithms with prescriptions for efficient update rules with inexpensive dynamic search length control, taking into account discretization effects and eliminating the need for ad-hoc penalty terms was proposed~\cite{Schirmer2011}. Comparizon, optimization and benchmarking of quantum control algorithms such as GRAPE methods with concurrent update of controls and Krotov-type methods with sequential update in a unifying programming framework was performed~\cite{Machnes2011}. Gradient and Hessian based algorithms such as GRAPE and BFGS were developed and widely applied for coherent control of closed and open quantum systems beyond initial proposal for design of NMR pulse sequences~\cite{khaneja_optimal_2005}. For example, a gradient-based optimal control algorithm was used to calculate the phase of BEC of ultra-cold $^87$Rb atoms in a one-dimensional optical lattice which optimally modifies the optical lattice and brings the quantum system to the desired target state~\cite{DupontPRX2021}. Combining gradient-based methods of quantum optimal control with automatic differentiation was performed~\cite{Goerz2022}. GRAPE with feedback was studied~\cite{Porotti2022}. Optimization of controlled-Z gate with GRAPE for superconducting qubits was performed~\cite{ZongPRApplied2021}. Higher-order Trotter decomposition for improving GRAPE was proposed~\cite{Yang2022}. Gradient-based optimization with quantum trajectories and automation differentiation was studied~\cite{PhysRevA.99.052327}.
Quantum control landscape analysis for single qubit phase-shift quantum gates using GRAPE  was recently performed~\cite{Volkov_2021_215303}.

A general method of incoherent control using this spectral density of the environment, including also in combination with coherent control, either subsequent or simultaneous, was developed and studied for multi-level quantum systems in~\cite{Pechen_Rabitz_2006}. In incoherent control, spectral density of the environment, which defines the distribution of its particles in momenta and internal degrees of freedom, is a control used to manipulate the system. This spectral density can sometimes be considered as  thermal (i.e., Planck distribution), but in general it can be any non-equilibrium non-negative function of momenta and internal degrees of freedom of environmental particles, which can even depend on time. It should be non-negative because its physical meaning is the density of particles. A natural example of such environment is incoherent photons interacting with a quantum system. In this case, the incoherent control is a generally non-equilibrium distribution functions $n_{\omega,\alpha}(t)$, where $\omega$ is photon frequency and $\alpha$ is polarization. Such incoherent control by time-dependent temperature can be realized experimentally. For example, implementing fast and controlled temperature variations for non-equilibrium control of thermal and mechanical changes in a levitated system is provided in~\cite{Rademacher_2022}. Another approach to incoherent control is to use back-action of non-selective quantum measurements to manipulate the quantum system, as was proposed~\cite{PechenIlinShuangRabitz2006} and studied, e.g., in~\cite{ShuangPechenHoRabitz2007,PechenTrushechkin2015}.

An advantage of incoherent control is that it allows, in combination with coherent control, to approximately steer {\it any} initial density matrix to {\it any} given target density matrix~\cite{PechenPRA2011}. This property approximately realizes controllability of open quantum systems in the set of all density matrices --- the strongest possible degree of quantum state control. This result has several important features. (1) It was obtained with a physical class of GKSL master equations well known in quantum optics and derived in the weak coupling limit. (2) It was obtained for almost all values of the physical parameters of this class of master equations and for multi-level quantum systems of an arbitrary dimension. (3) For incoherent controls in this scheme in some simple version an explicit analytic solution (not numerical) was obtained. (4) The scheme is robust to variations of the initial state --- the optimal control steers simultaneously {\it all} initial states into the target state, thereby physically realizing all-to-one Kraus maps previously theoretically exploited for quantum control~\cite{Wu_2007_5681}. However, after the existence of optimal controls is established, the next question is how to find optimal coherent and incoherent controls. This paper is devoted to answering this question.

For closed $N$--level quantum systems, their dynamics is represented by points on a (special) unitary group $SU(N)$. For open quantum systems, a formulation of completely positive trace preserving dynamics as points of \textit{complex Stiefel manifold} was proposed and theory of open system's quantum control as gradient flow optimization over suitable complex Stiefel manifolds (strictly speaking, a factor of Stiefel manifolds over some set of transformations) was developed in details for two-level~\cite{PechenJPA2008} and general quantum systems~\cite{OzaJPA2009} and applied for optimization of various physical quantities. This picture, where the dynamics is represented by points on the unitary group for closed systems, or on the Stiefel manifold for open systems, is called kinematic picture for quantum control. A unified framework and kinematic analysis for {\it open and closed loop quantum control} was provided in~\cite{PechenBrifPRA2010} and for {\it open classical and quantum systems} in~\cite{PechenRabitzEPL2010}. The picture, where the dynamics is described by Schr\"odinger or master equation, is called dynamic picture. 

For open systems with incoherent control, gradient- and Hessian-based optimization theory was developed only in the kinematic picture~\cite{PechenJPA2008,OzaJPA2009,PechenBrifPRA2010,PechenRabitzEPL2010}. In this work, we develop the GRAPE method for open quantum systems driven by both coherent and incoherent controls in the dynamic picture. The Mayer-type control problems are considered, where the objective functional depends on the terminal density matrix of the quantum system. An explicit expression for gradient of the objective functional is obtained for piecewise constant control. This computed analytical expression is then implemented for numerical optimization for the problem of steering a given initial density matrix of a two-level quantum system to a target final density matrix. For this case, surprisingly, an explicit, complicated by completely analytical, expression for gradient is obtained through diagonalization of a $3\times 3$ matrix determining the controlled system dynamics in the Bloch representation. For that we solved characteristic equation which is third order algebraic equation and solved by Cardano's formula~\cite{Cardano_formula}. This gives us exact expression for evolution of the system, otherwise we would have to use numeric methods. The efficiency of the algorithm is demonstrated through numerical simulations in the context of the state-to-state transition problem. The results also confirm recent finding of unreachable states in the Bloch ball of all states for a two-level system~\cite{Lokutsievskiy_2021}. Important is that the obtained analytical expressions for the gradient can be used in any method which relies on gradient (including for Hessian approximation), such as GRAPE and BFGS algorithms. Numerical simulations using expressions for the gradient and the evolution for one qubit were performed in~\cite{PetruhanovPhotonics2022} for the unitary gate generation problem, but without analytical analysis as it is done in this paper. Some problems for the two-qubit case were analyzed numerically in~\cite{PetruhanovIntJModA2022}. 

The structure of the paper is the following. In Sec.~\ref{Sec2} we consider master equation for a general $N$-level open quantum system driven by coherent and incoherent controls. In Sec.~\ref{Sec:IC}, we remind the basics of incoherent control by spectral density of the environment. In Sec.~\ref{Sec3}, main objective functionals and their derivatives with respect to the final state are presented. In Sec.~\ref{Sec5}, we calculate gradient of the objective for piecewise constant controls. In Sec.~\ref{Sec6}, the single qubit case is studied in details and complete analytical solution for evolution and gradient is obtained. In Sec.~\ref{Sec7}, numerical simulations for the state-to-state transfer are provided. Conclusions Sec.~\ref{Sec8} summarizes the work.

\section{Master equation for a quantum system driven by coherent and incoherent controls}\label{Sec2}

In this work, we consider an $N$-level open quantum system driven by coherent control $u = (u_k)_{k=1}^K$ and incoherent control $n = (n_{ij})_{1\leq i < j \leq N}$, whose density matrix evolution satisfies the GKSL master equation with control-dependent coefficients, as was introduced in~\cite{Pechen_Rabitz_2006}:
\begin{equation}
    \der{\rho_t}{t} = \mathcal{L}_t(\rho_t) = -i [H_0 + \displaystyle \sum_{k} u_k(t) V_k, \rho_t] + \L_{n(t)}(\rho_t),\qquad \rho_{t = 0} = \rho_0\,.
    \label{system}
\end{equation}
This master equation can be written as master equation with coherent control Hamniltonian $H_c(t)$ and time dependent decoherence rates $\gamma_k(t)$ (see page 3 of~\cite{Pechen_Rabitz_2006}) as
\begin{equation}\label{Eq:ME2}
\der{\rho_t}{t} = -i [H_0+H_c(t), \rho_t] + \sum\gamma_k(t) {\cal D}_k(\rho_t),
\end{equation}
The time dependent decoherence rates can be expressed via $n(t)$ (see page 3 of~\cite{Pechen_Rabitz_2006}).

In~(\ref{system}), $H_0$ is the free system Hamiltonian with eigenstates $|i\rangle$ ($i=1,\dots,N$), $V_k$ is the Hamiltonian of interaction of the system with $k$-th coherent control with values $u_k \in \R$, and superoperator $\mathcal{L}(.)_{n(t)}$ describes interaction of the system with the environment and depends on the incoherent control $n(t) = \{n_{ij}(t),\;1\le i<j\le N\}$. Important special property of the incoherent control is that by its physical meaning it takes only non-negative values, $n_{ij}(t) \ge 0$. Explicit form of $\L_{n(t)}$ is determined by the details of the physical model of the environment. There are two known forms of $\mathcal{L}_{n(t)}$ given by two different Markovian limits: weak coupling limit \cite{Spohn_1978} and low denstiy limit \cite{Dumcke}. If the environment is formed by incoherent radiation then the explicit form of $\mathcal{L}(.)_{n(t)}$ can be obtained in the weak coupling limit \cite{PechenPRA2011} and the superoperator $\L(\rho_t)$ can be written as follows: 
$$\mathcal{L} = \mathcal{L}_0 + \displaystyle \sum_k u_k(t) \mathcal{L}_k + \sum_{i<j} n_{ij}(t)\mathcal{L}_{ij},$$
where $\mathcal{L}_k = -i[V_k, \,\cdot\;]$. In this case, $\L_0$ and $\L_{ij}$ have the form~\cite{PechenPRA2011}:
\begin{align*}
    \L_0 &= -i[H_0, \cdot] + \displaystyle \sum_{i<j} A_{ij} \M_{ij}, \\
    \L_{ij} &= A_{ij}(\M_{ij} + \M_{ji}).
\end{align*}
Here $A_{ij}\ge 0$ are the Einstein coefficients for the transitons between stated $|i\rangle$ and $|j\rangle$ and
$$
\M_{ij}(\rho) = 2\ket{i}\bra{j}\rho\ket{j}\bra{i} - \ket{i}\bra{i}\rho - \rho\ket{i}\bra{i}.
$$
We denote by $\rho_t^{u,n}=\rho_t[u,n]$ the solution of the master equation~(\ref{system}) at time $t$ for the system driven by controls $u$ and $n$.

\section{Incoherent~Control}\label{Sec:IC}

Incoherent control as proposed in~\cite{Pechen_Rabitz_2006} is realized by using spectral density of the environment, which represents the state of the environment. This state can be thermal in the simplest case or non-equilibrium and time-dependent in general. Physically incoherent control can be implemented by using variable time-dependent spectral density of incoherent photons or phonons $n_\omega(t)$. This density is considered as a function of both transition frequency $\omega$ and time $t$, where frequency is indexed as a subscript for convinience. Of course, not for all transition frequencies this spectral density can  easily be experimentally made arbitrary, so practical applicability depends on the exact type of the environment (photons or phonons), spectrum of the system, and experimentally available techniques to generate various $n_\omega$. For photons, their spectral density generally should also include dependence on polarization, so that incoherent control should be a function of $\omega$, $\alpha$, and~$t$. In this work, we do not use this option and leave it for a future work. For~thermal light or black-body radiation, distribution of its photons has the famous time-independent Planck form
\[
n_\omega= \frac{1}{\pi^2}\dfrac{\omega ^3}{\exp{(\beta \omega)} - 1},
\]
where $\beta$ is the inverse temperature (in the Planck system of units, where reduced Planck constant $\hbar$, speed of light $c$, and~Boltzmann constant $k_\textrm{B}$ equal one). If temperature is varied, so that $\beta=\beta(t)$, then spectral density becomes time-dependent,
\[
n_\omega(t)= \frac{1}{\pi^2}\dfrac{\omega ^3}{\exp{(\beta (t) \omega)} - 1},
\]
However, variation of temperature is only the simplest option to vary $n_\omega(t)$ and generally spectral density can be made non-thermal with other shapes $n_\omega(t)$. Such non-thermal shapes can be obtained, e.g.,~by filtering black-body radiation, using light-emission diodes, etc, see Figure~\ref{Fig1}.

\begin{figure}[ht!]
\centering
\includegraphics[width = 0.6\linewidth]{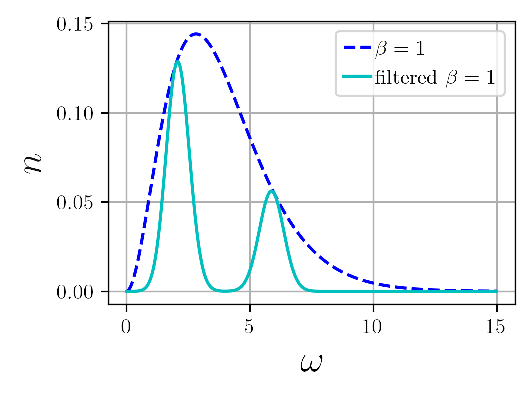}
\caption{Planck density of black-body radiation for $\beta = 1$ and its Gaussian filtering at $\omega=2$ and $\omega = 6$ with a variance $\sigma^2=0.5$.}
    \label{Fig1}
\end{figure}
Incoherent control has physical meaning of density of particles of the environment, and hence mathematically it should be a non-negative quantity, $n_\omega(t)\ge 0$. In~addition, one can require that  total density of photons at any time moment is finite, so that
\[
\int\limits_0^\infty n_\omega(t)d\omega<\infty\qquad \forall t \geq 0
\]

This incoherent control naturally makes decoherence rates time dependent. Time-dependent decoherence rates for transition between levels $i$ and $j$ with transton frequency $\omega_{ij}$ are then 
\[
\gamma_{ij}(t)=\pi\int d{\bf k}\delta(\omega_{ij}-\omega_{\bf k})|g({\bf k})|^2(n_{\omega_{ij}}(t)+\kappa_{ij})
\]
where $\kappa_{ij}=1$ for $i>j$ and $\kappa_{ij}=0$ otherwise, $\omega_{\bf k}$ is the dispersion law for the bath (e.g., $\omega=|{\bf k}|c$ for photons, where ${\bf k}$ is photon momentum and $c$ is the speed of light), and $g({\bf k})$ describes coupling of the esystem to $\bf k$-th mode of the reservoir. For $i>j$, $\kappa_{ij}=1$ corresponds to spontaneous emission and $\gamma_{ij}$ determines rate of spontaneous and induced emission between levels $i$ and $j$. For $i<j$, $\gamma_{ij}$ determines rate of induced absorption. These decoherence rates appear in Eq.~(\ref{Eq:ME2}), where $k=\{i,j\}$ is multi-index. 
The magnitude of the decoherence rates affects the speed of decay of off-diagonal elements of the system density matrix and determines the values of the diagonal elements towards which the diagonal part of the system density matrix evolves. This property can be used, for~example, for~approximate generation of various density matrices~\cite{PechenPRA2011}. 

\section{Quantum control objective functionals}\label{Sec3}
We consider as control goal the problem of optimizing an arbitrary Mayer-type objective functional, i.e., the objective functional which is completely determined by the system state at some final time $T>0$:
\begin{equation}
    \mathcal{J}(\rho_T^{u, n}) = \mathcal{F}[u,n].
    \label{functional}
\end{equation}
Some key examples of such quantum control objectives include:
\begin{itemize}
\item Maximization of the average expectation value of a system observable at the final time,
\begin{equation}
    \J_1(\rho_T^{u,n}) = \mathrm{Tr}(\rho_T^{u,n} \mathcal{O}) = \langle \rho_T^{u,n}, \mathcal{O} \rangle_\mathrm{HS} \to \max_{u,n},
\end{equation}
where $\langle \cdot , \cdot \rangle_\mathrm{HS}$ is the Hilbert-Schmidt scalar product of two matrices;
\item State-to-state transfer objective functional,
\begin{equation}
    \J_2(\rho_T^{u,n}) = \|\rho_T^{u,n} - \rho_\mathrm{target}\|^2 \to \min_{u,n},
    \label{functional_distance}
\end{equation}
where $\|\cdot\|$ is some matrix norm, which is used for finding control $(u,n)$ steering the system from the initial state $\rho_0$ to the target state $\rho_T^{u,n} = \rho_\mathrm{target}$;
\item Uhlmann--Jozsa fidelity~\cite{Mendonca_2008}, which can also be used for state-to-state transfer problem:
    \begin{eqnarray}
    \J_{\mathrm{UJ}}(\rho_T; \rho_{\rm target}) = 
    \left({\rm Tr}\sqrt{\sqrt{\rho_T} \rho_{\rm target} \sqrt{\rho_T}} \right)^2 \in [0, 1].
    \label{Uhlmann_Jozsa_fidelity}
    \end{eqnarray}
    One has $\J_{\mathrm{UJ}}(\rho_t; \rho_{\rm target}) = 1 \Leftrightarrow \rho_t = \rho_{\rm target}$.
\end{itemize}

Implementation of gradient methods requires computation of gradient of the objective functional with respect to coherent and incoherent controls:
\begin{equation}
{\rm grad}\,\F=\left(\deltapd{\mathcal{F}}{u}, \deltapd{\mathcal{F}}{n} \right).
    \label{functional_derivative}
\end{equation}
Here variational derivatives are Fréchet derivatives~\cite{Kolmogorov}. Because the functional is of Mayer-type, its gradient via chain rule for derivatives is determined by variation of the objective with respect to the final state and variation of the final state with respect to the controls:
$$
\deltapd{\F}{[u,n]} = \deltapd{\J}{\rho_T} \deltapd{\rho_T}{[u,n]}.
$$
\begin{theorem}
    Explicit expressions of 
    derivatives with respect to $\rho_T$ for the objective functionals described above can easily be computed and have the form:
    \begin{itemize}
    \item For maximization of the average expectation of a quantum observable,
    \begin{equation}
    \deltapd{\J_1}{\rho} = \langle \mathcal{O}, \; \cdot \;\rangle_\mathrm{HS}.
    \end{equation}
    \item For state-to-state transfer, if the norm $\|\cdot\|$ is induced by Hilbert-Schmidt scalar product then
    \begin{equation*}
\deltapd{\J_2}{\rho} = 2\left\langle \rho - \rho_\mathrm{target} , \quad\cdot\quad \right\rangle_\mathrm{HS}.
\end{equation*}
    \item For Uhlmann--Jozsa fidelity for a single qubit,
    \begin{equation}
        \dfrac{\partial}{\partial \rho} \J_{\mathrm{UJ}}(\rho,\rho_\mathrm{target}) = {\rm Tr}\left(\rho_\mathrm{target}\;\;\cdot\;\;\right) + \sqrt{|\rho_\mathrm{target}|}\sqrt{|\rho|} \mathrm{Tr}\left({\rho}^{-1}\;\; \cdot \;\;\right)\,.
    \end{equation}
    \end{itemize}
\end{theorem}

The remaining problem hence is to find variation of the final state with respect to the controls, $\delta\rho_T/\delta [u,n]$.

\section{Gradient of the objective for piecewise constant controls}\label{Sec5}

Computing gradient of any Mayer-type functional requires computing variation of the final state $\rho_T^f$ of the system. To compute this variation, below we consider the controls $u(t)$ and $n(t)$ as piecewise constant functions. Without loss of generality they can be written as follows:
\begin{equation}
    u_k(t) = \displaystyle \sum_{m=1}^M u_k^m \chi_{[t_{m-1}, t_m)}(t), \quad n_{ij}(t) = \displaystyle \sum_{m=1}^M n_{ij}^m \chi_{[t_{m-1}, t_m)}(t),
    \label{N_control}
\end{equation}
$$ 0 < t_0 < t_1 < \dots < t_M = T ,$$
where $\chi_D(t)$ is the indicator function of the set $D$. The evolution of the system driven by piecewise constant coherent and incoherent control (\ref{N_control}) exists and is unique in the space of piecewise smooth solutions. Moreover, it is given by composition of matrix exponentials:
\begin{equation} 
    \rho_T[u,n] = e^{\Delta t_M \mathcal{L}^M}\dots e^{\Delta t_1 \mathcal{L}^1} \rho_0,
    \label{N_final_state}
\end{equation}
where
$$\mathcal{L}^m \equiv \mathcal{L}(t_m), \quad \Delta t_m = t_m - t_{m-1},\quad m = 1,\dotsc,M.$$
Incoherent control $n(t)$ takes non-negative values and therefore is constrained. For adopting gradient optimization method we introduce new control functions $w_{ij}(t)$ such that:
\begin{equation}
    n_{ij} = (w_{ij})^2.
    \label{new_incoherent_control}
\end{equation}
Now the control space $\R^{\mathrm{dim}u \cdot M} \times \R^{\frac{N(N-1)}{2} \cdot M}$ is not constrained.
\begin{theorem}\label{theoremGRAPE_PC}
If control is piecewise constant and has the form~(\ref{N_control}), then derivative of the final state $\rho_T$ with respect to the components of control~$(u, w)$ is given by:
\begin{align*}
    \pd{}{u_k^m} \rho_T &= e^{\Delta t_M \mathcal{L}^M}\dots e^{\Delta t_{m+1} \mathcal{L}^{m+1}} \left(\pd{}{u_k^m}e^{\Delta t_m \mathcal{L}^m}\right) e^{\Delta t_{m-1} \mathcal{L}^{m-1}}\dots e^{\Delta t_1 \mathcal{L}^1} \rho_0,\\
    \pd{}{w_{ij}^m} \rho_T &= e^{\Delta t_M \mathcal{L}^M}\dots e^{\Delta t_{m+1} \mathcal{L}^{m+1}} \left(\pd{}{w_{ij}^m}e^{\Delta t_m \mathcal{L}^m}\right) e^{\Delta t_{m-1} \mathcal{L}^{m-1}}\dots e^{\Delta t_1 \mathcal{L}^1} \rho_0,
\end{align*}
where
\begin{align*}
    \pd{}{u_k^m}e^{\Delta t_m \mathcal{L}^m} &= \Delta t_m \displaystyle \int_0^1 \exp(\alpha \Delta t_m \mathcal{L}^m) \mathcal{L}_k^m \exp((1 - \alpha) \Delta t_m \mathcal{L}^m)\,\d\alpha,\\
    \pd{}{w_{ij}^m}e^{\Delta t_m \mathcal{L}^m} &= 2 w_{ij}^m \Delta t_m \displaystyle \int_0^1 \exp(\alpha \Delta t_m \mathcal{L}^m) \mathcal{L}_{ij}^m \exp((1 - \alpha) \Delta t_m \mathcal{L}^m)\,\d\alpha.
\end{align*}
\end{theorem}

\begin{proof} Partial derivative of $\rho_T$~(\ref{N_final_state}) with respect to control on $m$-th time segment acts on only one factor in the product, which is the matrix exponential describing the dynamics on this $m$-th time segment. The partial derivative of matrix exponential is given by the known integral formula \cite{Wilcox}:
\begin{equation}
    \der{}{x} e^{A(x)} = \int \limits _0^1 e^{\alpha A(x)} \der{A(x)}{x} e^{(1-\alpha)A(x)} d\alpha.
    \label{special_formula}
\end{equation}
\end{proof}

\section{Exact solution for a single qubit: evolution and objective gradient}
\label{Sec6}
\subsection{Dynamics in the Bloch ball}
The expressions in Theorem~\ref{theoremGRAPE_PC} are quite general. In this section, we explicitly consider and solve analytically the case of a two-level quantum system or qubit. The master equation for this system has the form \cite{PechenPRA2011}:
\begin{equation}
    \der{}{t} \rho_t = - i [H_0 + V u(t), \rho_t] + \gamma \mathcal{L}_{n(t)}(\rho_t),\quad \rho(0) = \rho_0.
    \label{system_qubit}
\end{equation}
Here the free Hamiltonian is
$H_0 = \omega \begin{pmatrix}
0 & 0 \\
0 & 1
\end{pmatrix},$
the control field interaction operator is
$V = \mu\sigma_x
= \mu\begin{pmatrix}
0 & 1 \\
1 & 0
\end{pmatrix}$, where $\mu>0$ is the dipole moment, $\gamma>0$ determines strength of interaction with the environment, the dissipative superoperator is
\begin{multline*} 
\mathcal{L}_{n(t)}(\rho_t) = n(t) \left(\sigma^+\rho_t\sigma^- + \sigma^-\rho_t\sigma^+ - \dfrac{1}{2}\{\sigma^-\sigma^+ + \sigma^+\sigma^- , \rho_t\}\right) \\ + \left(\sigma^+\rho_t\sigma^- -  \dfrac{1}{2}\{\sigma^-\sigma^+, \rho_t\}\right),\quad n(t) \geqslant 0;
\end{multline*}
matrices
$ 
\sigma^\pm = \dfrac{1}{2} (\sigma_x \pm i\sigma_y)$, $\sigma^+ = \begin{pmatrix}
0 & 1 \\
0 & 0
\end{pmatrix}$, $\sigma^-= \begin{pmatrix}
0 & 0 \\
1 & 0
\end{pmatrix}$, and
$\sigma_x$, $\sigma_y$, $\sigma_z$ are the Pauli matrices. Such master equation for a single qubit with time-dependent decoherence rate $\gamma(t)$ determined by $n(t)$ was considered in~\cite{PechenPRA2011,MorzhinPechenIJTP2021,Lokutsievskiy_2021, MorzhinLJM2020, Morzhin_Pechen_maximization}.

For the two-level system, we will use the  convenient parametrization of the density matrix $\rho_t$ by the Bloch vector $\mathbf{r}=(r_1,r_2,r_3)\in\mathbb R^3$:
\begin{equation}
    r_i = \mathrm{Tr}\rho \sigma_i,\quad i \in \{x, y, z\}, \quad  \rho_t = \dfrac{1}{2} (\mathbb{I} + (\mathbf{r}, \boldsymbol{\sigma})), \quad \boldsymbol{\sigma} = (\sigma_x, \sigma_y, \sigma_z),
    \label{bloch variables}
\end{equation}
where $\mathbb{I}$ is the $2\times2$ identity matrix. In this representation, we obtain the following system of equations for the dynamics of the Bloch vector $\r(t)$:
\begin{equation}
    \der{\r(t)}{t} = (B + B^u u(t) + B^n n(t) ) \mathbf{r}(t) + \mathbf{b},
    \label{bloch equation}
\end{equation}
where
$$ \quad 
B = \begin{pmatrix}
- \dfrac{\gamma}{2} & \omega & 0 \\
- \omega & - \dfrac{\gamma}{2} & 0 \\
0 & 0 & - \gamma \\
\end{pmatrix}, \quad B^u = 
\begin{pmatrix}
0 & 0 & 0 \\
0 & 0 & -2\mu \\
0 & 2\mu & 0 \\
\end{pmatrix}, $$
$$\quad B^n = 
\begin{pmatrix}
-\gamma & 0 & 0 \\
0 & - \gamma & 0 \\
0 & 0 & -2\gamma \\
\end{pmatrix}, \quad \mathbf{b} = 
\begin{pmatrix}
0 \\
0 \\
\gamma \\
\end{pmatrix}.$$
Denote the matrix in the right hand side (r.h.s.) of the linear equation~(\ref{bloch equation}) by
$$
A[u(t),n(t)] = B + B^u u(t) + B^n n(t).
$$

Knowledge of the gradient of the Bloch vector $\r(t)$ with respect to controls allows to compute gradient of the objective functional:
$$
\mathcal{J}(\rho_T[u,n]) = \mathcal{J}_{\rm Bloch}(\mathbf{r}_T[u,n]) = \mathcal{F}[u,n],
$$
using the chain rule as
\begin{equation}
    \deltapd{\F}{[u, n]} = \deltapd{\J}{\mathbf{r}_T} \deltapd{\mathbf{r}_T}{[u, n]}.
    \label{functional_gradient_bloch}
\end{equation}

\subsection{Explicit expression for the gradient with respect to PC control}

Now consider the controls $u(t)$ and $n(t) = w(t)^2$ as piecewise constant functions:
\begin{align}
    u(t) &= \displaystyle \sum_{k=1}^M u_k \chi_{[t_{k-1}, t_k)}(t), \qquad u_k\in\mathbb R\label{PC qubit control: u}\\
    w(t) &= \displaystyle \sum_{k=1}^M w_k \chi_{[t_{k-1}, t_k)}(t), \qquad w_k\in\mathbb R\label{PC qubit control: w}\\
    n(t) &= w(t)^2 =  \displaystyle \sum_{k=1}^M n_k \chi_{[t_{k-1}, t_k)}(t), \qquad n_k\in\mathbb R_+
    \label{PC qubit control: n}
\end{align}
$$ n_k = w_k^2,\quad 0 < t_0 < t_1 < \dots < t_M = T ,$$
Denote the matrix in the right-hand side (r.h.s.) of the equation (\ref{bloch equation}) on time segment $[t_{k-1}, t_k)$ by
$$A_k =A[u_k, n_k] = B + B^u u_k + B^n n_k.$$
For solution of the equation with piecewise constant r.h.s we can write the following recurrent formula and express each state vector of the sequence~$\{\r^k\}_{k = 0}^{M}$:
\begin{align}
    \mathbf{r}^k &\equiv \mathbf{r}(t_{k}) = e^{A_k \Delta t_k} \mathbf{r}^{k - 1} + \mathbf{g}_k
    \label{recurrent eq r}\\
     &= e^{A_k \Delta t_k} \dots e^{A_1 \Delta t_1} \mathbf{r}_0 + e^{A_k \Delta t_k} \dots e^{A_2 \Delta t_2} \mathbf{g}_1 + \dotsc + e^{A_k \Delta t_k} \mathbf{g}_{k-1} + \mathbf{g}_{k},
    \label{Bloch k_th state}
\end{align}
where
\begin{equation}
    \mathbf{g}_k = (e^{A_k \Delta t_k} - \mathbb{I}) A_k^{-1} \mathbf{b}, \quad \Delta t_k  = t_k - t_{k - 1}.
    \label{g definition}
\end{equation}

\begin{theorem}
If control is piecewise constant and given by~(\ref{PC qubit control: u}) and (\ref{PC qubit control: w}) then the derivatives of the final state $\mathbf{r}_T = \r^M$ with respect to the components of control~(u, w) are equal to
\begin{equation} 
    \pd{\mathbf{r}_T}{u_j}  = e^{A_N \Delta t_N} \dots e^{A_{j + 1} \Delta t_{j + 1}}\Bigg[\, \pd{}{u_j}\left(e^{A_j \Delta t_j}\right)\,\r_{j-1} + \pd {\mathbf{g}_j}{u_j}\Bigg],
    \label{final state gradient u}
\end{equation}
\begin{equation}
    \pd{\mathbf{r}_T}{w_j}  = e^{A_N \Delta t_N} \dots e^{A_{j + 1} \Delta t_{j + 1}}\Bigg[\, \pd{}{w_j}\left(e^{A_j \Delta t_j}\right)\,\r_{j-1} + \pd {\mathbf{g}_j}{w_j}\Bigg].
    \label{final state gradient w}
\end{equation}
Here derivatives are given by
\begin{align}
    \pd {\mathbf{g}_k}{u_k} &= \left(\pd{}{u_k}e^{A_k \Delta t_k} - (e^{A_k \Delta t_k} - \I)A_k^{-1} B^u\right) A_k^{-1} \mathbf{b},\label{g gradient u}\\ 
    \pd{}{u_k}e^{A_k \Delta t_k} &= \Delta t_k\int \limits _0^1 \mathrm{exp}({\alpha A_k \Delta t_k })\, B^u\, \mathrm{exp}({(1-\alpha)A_k \Delta t_k}) \d\alpha,\label{F gradient u}\\
    \pd {\mathbf{g}_k}{w_k} &= \left(\pd{}{w_k}e^{A_k \Delta t_k} - 2 w_k (e^{A_k \Delta t_k} - \I)A_k^{-1} B^n\right) A_k^{-1} \mathbf{b},\label{g gradient w}\\ 
    \pd{}{w_k}e^{A_k \Delta t_k} &= 2 w_k \Delta t_k\int \limits _0^1 \mathrm{exp}({\alpha A_k \Delta t_k })\, B^n\, \mathrm{exp}({(1-\alpha)A_k \Delta t_k}) \d\alpha.\label{F gradient w}
\end{align}
\end{theorem} 

\begin{proof} Expressions~(\ref{final state gradient u}) and~(\ref{final state gradient w}) can be obtained by differentiating expression for Bloch vector state $\r_k$~(\ref{Bloch k_th state}) for $k = M$ with respect to control components $u_j$ and $w_j$. 
Expressions~(\ref{g gradient u}) and~(\ref{g gradient w}) for derivatives $\pd {\mathbf{g}_k}{u_k}$, $\pd {\mathbf{g}_k}{w_k}$ are given by differentiating~(\ref{g definition}).
In $A_j = B + B^u u_j  + B^n n_j$, the matrix $B^u$ does not commute with other summands (the same is true also for $B^n$), so the derivatives $\pd{}{u_k}e^{A_k \Delta t_k}$, $\pd{}{w_k}e^{A_k \Delta t_k}$ can be found using the special formula (\ref{special_formula}).
\end{proof}

\subsection{Diagonalization of the matrix exponential}
\label{Diagonalization}

Gradient-based optimization methods is especially effective when analytical expressions for the gradient can be obtained.  For a single qubit, it can be done via computation of the matrix exponential $e^{A_k \Delta t_k}$. In order to compute this exponent, the matrix $A_k$ should be diagonalized via finding its eigenvalues and eigenvectors. 

Consider the following matrix, which depends on the control $(u, n)$ and equals to $\frac{1}{\omega}A_k$:
\begin{equation}
    \bar{A}[u,n] = \frac{1}{\omega}A_k=\begin{pmatrix}
    -\dfrac{\gamma}{2\omega} - \dfrac{\gamma}{\omega} n && 1 && 0 \\[0.7em]
    -1 && -\dfrac{\gamma}{2\omega} - \dfrac{\gamma}{\omega} n &&    -\dfrac{2\mu}{\omega} u \\[1em]
    0 &&  \dfrac{2\mu}{\omega} u && -\dfrac{\gamma}{\omega} -  \dfrac{2\gamma}{\omega} n
    \label{matrix_to_diag}
    \end{pmatrix}.
\end{equation}
For the two-level case, the characteristic equation \,$\mathrm{det}(\bar{A} - \lambda \I) = 0$\, is the cubic equation:
\begin{equation}
    \lambda^3 + 4 \hat{n} \lambda^2 + (5 \hat{n}^2 + \hat{u}^2 + 1) \lambda + \hat{u}^2 \hat{n} + 2 \hat{n}^3 + 2 \hat{n} = 0,
    \label{characterstic_equation}
\end{equation} 
where
\begin{equation}
    \hat{u} = \dfrac{2\mu u}{\omega},\quad \hat{n} = \dfrac{\gamma}{\omega} \left(n + \dfrac{1}{2}\right).
    \label{control_hat}
\end{equation}
This equation can be solved by implementing Cardano's method \cite{Cardano_formula}. Denote:
\begin{align*} 
    p &= \dfrac{\hat{u}^2}{3}             -\dfrac{\hat{n}^2}{9} + \dfrac{1}{3}, \\
    q &= - \dfrac{\hat{u}^2 \hat{n}}{6} + \dfrac{\hat{n}^3}{27} + \dfrac{ \hat{n}}{3}, \\
    \Delta &= p^3 + q^2 = 
     \dfrac{\hat{u}^6}{27} -\dfrac{\hat{u}^4 \hat{n}^2}{108}  + \dfrac{ \hat{u}^4}{9} - \dfrac{5 \hat{u}^2 \hat{n}^2}{27} + \\ &\qquad\qquad\qquad\qquad\qquad\qquad\dfrac{ \hat{n}^4}{27} + \dfrac{ \hat{u}^2}{9} + \dfrac{2 \hat{n}^2}{27}  + \dfrac{1}{27}.
\end{align*}
If $p \neq 0$, then roots of the characteristic equation (\ref{characterstic_equation}) are
\begin{equation}
    \lambda_k = -\dfrac{4 \hat{n}}{3} + \xi_k - \dfrac{p}{\xi_k},\quad k = 1,2,3,
    \label{eigenvalues_1}
\end{equation}
where $\xi_k$ are three cubic roots
\begin{equation*}
    \xi_k = \sqrt[3]{-q + \sqrt{\Delta}}.
\end{equation*}
For the root $\sqrt{\Delta}$,  any of the two roots can be chosen (for example, positive real or with positive imaginary part).

If $p = 0$, then the roots of the characteristic equation (\ref{characterstic_equation}) are:
\begin{equation}
    \lambda_k =  -\dfrac{4 \hat{n}}{3} + \xi_k,\quad k = 1,2,3,
    \label{eigenvalues_2}
\end{equation}
where
$$\xi_k = \sqrt[3]{-2q}.$$

Let us now find eigenvectors of the matrix (\ref{matrix_to_diag}). The matrix
\begin{equation}
    \bar{A} - \lambda \I = \begin{pmatrix}
-\hat{n} - \lambda & 1 & 0\\
-1 & -\hat{n} - \lambda & -\hat{u} \\
0 & \hat{u} & -2\hat{n}-\lambda
\end{pmatrix}
\label{matrix A minus lambda I}
\end{equation}
has rank $2$ if $1 + (n+\lambda)^2 \neq 0$, i.e. $\lambda \neq -n \pm i$. Then each eigenvalue has one-dimensional eigenspace and corresponding eigenvector is
\begin{equation}
    h_\lambda = \begin{pmatrix}
    \dfrac{\hat{u}}{1 + (\hat{n} + \lambda)^2} \\[0.9em]
    \dfrac{\hat{u}(\hat{n} + \lambda)}{1 + (\hat{n} + \lambda)^2}\\[0.8em]
    -1
    \end{pmatrix}.
    \label{eigenvector}
\end{equation}
The case $\lambda=-n \pm i$, when the denominator in the expression for $h_\lambda$ above is zero, will be considered in subsection~\ref{Sec:exact solution for hat u = 0} --- this case corresponds to complete absence of coherent control.

Otherwise, if there is a double or triple root then we might need to get a matrix in Jordan normal form instead of diagonal matrix. Corresponding Jordan chains are: $(h_\lambda, h'_\lambda)$ --- for double root, and $(h_\lambda,  h'_\lambda, h''_\lambda)$ --- for triple root, where
\begin{align}
    h'_\lambda &= \dfrac{\hat{u}}{[1 + (\hat{n} + \lambda)^2]^2}\begin{pmatrix}
    -2(\hat{n} + \lambda) \\[0.4em]
    1 - (\hat{n}+\lambda)^2\\[0.4em]
    0
    \end{pmatrix}, \label{jordan_chain:1} \\[0.4em]
    h''_\lambda &= \dfrac{\hat{u}}{[1 + (\hat{n} + \lambda)^2]^3}\begin{pmatrix}
    3(\hat{n} + \lambda)^2 - 1 \\[0.4em]
    ((\hat{n}+\lambda)^2 - 3)(\hat{n}+\lambda)\\[0.4em]
    0
    \end{pmatrix}.
    \label{jordan_chain:2}
\end{align}

Depending on multiplicities of the roots~(\ref{eigenvalues_1}) and~(\ref{eigenvalues_2}) there are three typical cases of roots' character of the characteristic equation~(\ref{characterstic_equation}) and they can be determined by sign of the coefficient $\Delta$:
\begin{enumerate}
    \item $\lambda_1$, $\lambda_2$, $\lambda_3$ --- three different eigenvalues (all real under the condition $\Delta < 0$ or one real and two complex conjugate under $\Delta > 0$);
    \item $\lambda_1$, $\lambda_2$ --- one eigenvalue of multiplicity 1 and one double eigenvalue if $\Delta = 0$ ($|p|^2 + |q|^2 \neq 0$);
    \item $\lambda$ --- one triple eigenvalue if $p = q = 0$ (and hence $\Delta = 0$).
\end{enumerate}

We diagonalize matrix $\bar{A}$~(\ref{matrix_to_diag}) and matrix exponential $e^{A[u,n] \Delta t}$ as follows:
\begin{equation}
    \bar{A} = S \Lambda S^{-1},\quad e^{A \Delta t} = e^{\bar{A}\omega \Delta t} = S e^{\Lambda\omega \Delta t} S^{-1},
    \label{matrix_exponential_diagonalized}
\end{equation}
where matrix $\Lambda$ and $S$ are different for each of three cases. Let us study this cases in more detail.

\subsubsection{The case $\Delta \neq 0$}
In the case $\Delta \neq 0$ there are three different eigenvalues~$\lambda_k$ and therefore there are three eigenvectors given by formula~(\ref{eigenvector}) of matrix~(\ref{matrix_to_diag}) for each eigenvalues. Thus in~(\ref{matrix_exponential_diagonalized}) 
\begin{equation}
    \Lambda = \mathrm{diag}(\lambda_1,\lambda_2,\lambda_3),
    \quad e^{\Lambda \omega \Delta t} = \mathrm{diag}(e^{\lambda_1 \omega \Delta t}, e^{\lambda_2 \omega \Delta t}, e^{\lambda_3 \omega \Delta t}),
    \quad S = (h_{\lambda_1},h_{\lambda_2},h_{\lambda_3}).
    \label{case Delta neq 0}
\end{equation}
Expression for the eigenvector~(\ref{eigenvector}) is obtained under the condition $1 + (n + \lambda)^2 \neq 0$. Expressions for the case $1 + (n + \lambda)^2 = 0$ are obtained in Section~\ref{Sec:exact solution for hat u = 0}.

\subsubsection{The case $\Delta = 0$ ($|p|^2 + |q|^2 \neq 0$)}
In the case $\Delta \neq 0$ we have one root of the characteristic equation~(\ref{characterstic_equation}) $\lambda_1$ with multiplicity 1 and one double root $\lambda_2$:
\begin{equation*}
    \lambda_1 = -\dfrac{4\hat{n}}{3},\quad \lambda_2 = -\dfrac{4\hat{n}}{3} + p.
\end{equation*}
Eigenvector~(\ref{eigenvector}) of $\lambda_1$ is
\begin{equation*}
    h_{\lambda_1} = \begin{pmatrix}
        \dfrac{4}{1 + \frac{\hat{n}^2}{9}} \\[1em]
        -\dfrac{\hat{u}\hat{n}}{3(1 + \frac{\hat{n}^2}{9})} \\[1em]
        -1
    \end{pmatrix}.
\end{equation*}
Eigenvalue $\lambda_2$ has one-dimensional space, therefore we have to use Jordan normal form to calculate the matrix exponential. Vectors of Jordan chain $(h_{\lambda_2}, h_{\lambda_2}')$~(\ref{eigenvector}) and (\ref{jordan_chain:1}) can be written as
\begin{equation*}
    h_{\lambda_2} = \begin{pmatrix}
        \dfrac{\hat{u}}{1 + \left(\frac{\hat{n}}{3} - p\right)^2} \\[1.2em] 
        \dfrac{\hat{u}\left(\frac{\hat{n}}{3} - p\right)}{1 + \left(\frac{\hat{n}}{3} - p\right)^2} \\[1em]
        -1
    \end{pmatrix},\quad h_{\lambda_2}' = \dfrac{1}{\left[1 + \left(\frac{\hat{n}}{3} - p\right)^2\right]^2}\begin{pmatrix}
        2\left(\frac{\hat{n}}{3} - p\right)\\[0.2em]
        1 - \left(\frac{\hat{n}}{3} - p\right)^2 \\[0.2em]
        0
    \end{pmatrix}.
\end{equation*}
Thus
\begin{equation}
    \Lambda = 
    \begin{pmatrix}
        \lambda_1 & 0 & 0\\
        0 & \lambda_2 & 1\\
        0 & 0 & \lambda_2
    \end{pmatrix},\quad e^{\Lambda \omega \Delta t} = \begin{pmatrix}
        e^{\lambda_1 \omega\Delta t} & 0 & 0 \\[0.2em]
        0 & e^{\lambda_2 \omega\Delta t} & \omega\Delta t \,e^{\lambda_2 \omega\Delta t} \\[0.2em]
        0 & 0 & e^{\lambda_2 \omega\Delta t}
    \end{pmatrix},
    \quad S = (h_{\lambda_1},h_{\lambda_2},h'_{\lambda_2}).
    \label{case Delta = 0}
\end{equation}

\subsubsection{The case $p = q = 0$}
The last case is the condition $p = q = 0$ which corresponds to triple root of characteristic equation~(\ref{characterstic_equation}).  In this case $\hat{u} = 2\sqrt{2}$, $\hat{n} = 3\sqrt{3}$ and the triple root is
\begin{equation*}
    \lambda = -\dfrac{4\hat{n}}{3} = -4\sqrt{3}.
\end{equation*}
This eigenvalue has only one eigenvector and therefore we have to use Jordan normal form to calculate the matrix exponential. Jordan chain vectors $(h_{\lambda_2}, h_{\lambda_2}')$~(\ref{eigenvector}),~(\ref{jordan_chain:1}), and~(\ref{jordan_chain:2}) are equal to
\begin{equation*}
    h_\lambda = \begin{pmatrix}
        \frac{1}{\sqrt{2}} \\[0.5em]
        -\frac{\sqrt{3}}{\sqrt{2}} \\[0.5em]
        -1
    \end{pmatrix},\quad 
    h_\lambda' = \begin{pmatrix}
        \frac{\sqrt{3}}{2\sqrt{2}} \\[0.5em]
        \frac{1}{2\sqrt{2}} \\[0.5em]
        0
    \end{pmatrix}, \quad 
    h_\lambda'' = \begin{pmatrix}
        \frac{1}{2\sqrt{2}} \\[0.5em]
        0 \\[0.5em]
        0
    \end{pmatrix}.
\end{equation*}
Then 
\begin{equation}
    \Lambda = 
    \begin{pmatrix}
        \lambda & 1 & 0\\
        0 & \lambda & 1\\
        0 & 0 & \lambda
    \end{pmatrix},
    \quad e^{\Lambda \omega \Delta t} = e^{\lambda \omega \Delta t}\begin{pmatrix}
            1 & \omega \Delta t & \dfrac{(\omega \Delta t)^2}{2} \\[0.5em]
            0 & 1 & \omega \Delta t \\
            0 & 0 & 1
        \end{pmatrix},
    \quad S = (h_\lambda,h'_\lambda,h''_\lambda).
    \label{case Delta p = q = 0}
    \end{equation}

\subsection{Exact solution for the case $u=0$}
\label{Sec:exact solution for hat u = 0}

As a particular case, we consider $u=0$ (in this case $\hat{u} = 0$) so that only incoherent control is applied to the system.

\begin{proposition}
Coherent control (\ref{control_hat}) equals zero $\hat{u} = 0$ if and only if characteristic equation (\ref{characterstic_equation}) has roots $\lambda = - n \pm i$:
$$\hat{u} = 0 \Longleftrightarrow \lambda = - n \pm i.$$
\end{proposition}

\begin{proof}
    Firstly, if $\hat{u} = 0$ then one can use obtained expression~(\ref{eigenvalues_1}) given by Cardano method, or we can just consider matrix (\ref{matrix A minus lambda I}). Then one real root is already factorized:
    $$
    \det (A - \lambda \I) = -((\hat{n} + \lambda)^2 + 1)(2\hat{n} + \lambda) = 0.
    $$
    Thus one root is $\lambda = -2\hat{n}$ and other two roots are complex conjugate $\lambda = -n \pm i$.
    
    Conversely, let characteristic equation~(\ref{characterstic_equation}) have two complex conjugate roots $\lambda = -n \pm i$. We need to show that then necessarily $\hat{u} = 0$. For that we can use Vieta's formulas. Implementing Vieta's formula to the coefficient at $\lambda^2$:
    \begin{align*}
        \lambda_1 + \lambda_2 + \lambda_3 &= -4\hat{n},\\
        -2\hat{n} + \lambda_3 &= -4\hat{n}
    \end{align*}
    we get that the third root is $\lambda = - 2\hat{n}$. 
    Now write Vieta's formula for the coefficient at $\lambda$:
    \begin{align*}
        \lambda_1 \lambda_2 + \lambda_1 \lambda_3 + \lambda_2 \lambda_3 &= 5\hat{n}^2 + \hat{u}^2 + 1,\\
        5\hat{n}^2 + 1 &= 5\hat{n}^2 + \hat{u}^2 + 1,\\
        0 &= \hat{u}^2.
    \end{align*}
    Thus we get that $\hat{u} = 0.$
    This finishes the proof.
\end{proof}

From this proved proposition it follows that the condition $1 + (n+\lambda)^2 \neq 0$ from Section~(\ref{Diagonalization}) is equivalent to $\hat{u} \neq 0$. Let us cover now the remaining case $\hat{u} = 0$.

As was shown earlier, $\hat{u} = 0$ implies that the roots of the characteristic equation~(\ref{characterstic_equation}) are
$$\lambda_{1,2} = -\hat{n} \pm i, \quad \lambda_3 = -2\hat{n}.$$
Corresponding eigenvectors are found as
$$h_{1,2} = \begin{pmatrix}
    1\\
    \pm i\\
    0
\end{pmatrix},\quad h_3 = \begin{pmatrix}
    0\\
    0\\
    1
\end{pmatrix}.$$
Then the matrix exponential~(\ref{matrix_exponential_diagonalized}) takes comprehensible form:
$$e^{A\Delta t} = \begin{pmatrix}
    e^{-\hat{n}\omega\Delta t}\cos{(\hat{n}\omega\Delta t)} & e^{-\hat{n}\omega\Delta t}\sin{(\hat{n}\omega\Delta t)} & 0 \\
    -e^{-\hat{n}\omega\Delta t}\sin{(\hat{n}\omega\Delta t)} & e^{-\hat{n}\omega\Delta t}\cos{(\hat{n}\omega\Delta t)} & 0 \\
    0 & 0 & e^{-2\hat{n}\omega\Delta t}.
\end{pmatrix}$$
Thus it simply describes an exponentially decaying rotation in the $xy$-plane and exponential decay along $z$-axis in case $\hat{n} > 0$.

Then exact evolution is as follows:
\begin{equation}
    \r = \begin{pmatrix}
    e^{-\hat{n}\omega\Delta t}\cos{(\hat{n}\omega\Delta t)} & e^{-\hat{n}\omega\Delta t}\sin{(\hat{n}\omega\Delta t)} & 0 \\
    -e^{-\hat{n}\omega\Delta t}\sin{(\hat{n}\omega\Delta t)} & e^{-\hat{n}\omega\Delta t}\cos{(\hat{n}\omega\Delta t)} & 0 \\
    0 & 0 & e^{-2\hat{n}\omega\Delta t}
    \end{pmatrix} \r_0 +
    \begin{pmatrix}
        0 \\
        0 \\
        \dfrac{1 - e^{-2\hat{n}\omega\Delta t}}{2n + 1}
    \end{pmatrix}.
    \label{exact evolution for hat u = 0}
\end{equation}

\subsection{Exact solution for the case $\gamma=0$}

The condition $\hat \gamma = 0$ corresponds to the case when the system~(\ref{system_qubit}) is closed, i.e., isolated from the environment.
In this case $\hat{n} = 0$ and the eigenvalues~(\ref{eigenvalues_1}) and~(\ref{eigenvalues_2}) take a much simplier form:
\begin{equation}
    \lambda = 0,\;\pm i\sqrt{1 + u^2}.
\end{equation}
The corresponding eigenvectors are equal to
\begin{equation}
    h_0 = \begin{pmatrix}
        \hat{u}\\
        0\\
        -1
    \end{pmatrix},\quad h_{\pm i\sqrt{1 + u^2}} = \begin{pmatrix}
        -\dfrac{1}{\hat{u}}\\[0.8em]
        \mp i\dfrac{\sqrt{1 + u^2}}{\hat{u}}\\[0.8em]
        -1
    \end{pmatrix}
\end{equation}
This allows to obtain for the matrix exponential the expression:
\begin{equation}
    e^{A\Delta t} = 
    \begin{pmatrix}
        \dfrac{\hat{u}^2 + \cos{(\bar{\omega}\Delta t)}}{1 + \hat{u}^2} & 
        \dfrac{\sin{(\bar{\omega}\Delta t)}}{\sqrt{1 + \hat{u}^2}} & 
        \dfrac{\hat{u}}{1 + \hat{u}^2}(-1 + \cos{(\bar{\omega}\Delta t)}) \\
        - \dfrac{\sin{(\bar{\omega}\Delta t)}}{\sqrt{1 + \hat{u}^2}} &
        \cos{(\bar{\omega}\Delta t)} &
        -\dfrac{\hat{u}}{\sqrt{1 + \hat{u}^2}}\sin{(\bar{\omega}\Delta t)} \\
        \dfrac{\hat{u}}{1 + \hat{u}^2}(-1 + \cos{(\bar{\omega}\Delta t)}) &
        \dfrac{\hat{u}}{\sqrt{1 + \hat{u}^2}}\sin{(\bar{\omega}\Delta t)} &
        \dfrac{1 + \hat{u}^2\cos{(\bar{\omega}\Delta t)}}{1 + \hat{u}^2},
        \label{matrix_exponential_gamma=0}
    \end{pmatrix}
\end{equation}
where $\bar{\omega} =\omega \sqrt{1 + \hat{u}^2}$. The corresponding evolution of the Bloch vector is
\begin{equation}
    \r = e^{A\Delta t} \r_0.
\end{equation}

Now we compare these expressions for the case $\gamma=0$ with analytical expressions for unitary evolution matrix in closed quantum system given in~\cite{Volkov_2021_215303}. In that work Hamiltonian has the form~$H = \sigma_z + a \sigma_y$ which is equivalent to our system with $\gamma = 0$ (closed system), $\omega=-2$, $\mu = 1$, and control $u = a$. Then $\hat{u} = {2\mu u}/{\omega} = -a$ and $\bar{\omega} \Delta t = \omega\sqrt{1 + \hat{u}^2} \Delta t = -2\alpha $. Closed system dynamics is fully defined by the unitary evolution matrix~$U$ and for constant coherent control~$a$ is~\cite{Volkov_2021_215303}:
\begin{equation*}
    U=\cos\alpha-i\Delta t(\sigma_z+a\sigma_x)\frac{\sin\alpha}{\alpha},\quad \alpha=\Delta t\sqrt{1+a^2},
\end{equation*}
The corresponding evolution of the Bloch vector is:
\begin{equation*}
    r_i = \mathrm{Tr}(\rho \sigma_i) = \mathrm{Tr} (U\rho_0U^\dagger \sigma_i) = \mathrm{Tr}(U^\dagger\sigma_iU\rho_0).
\end{equation*}
Decomposing $U^\dagger\sigma_iU = \displaystyle\sum_{j \in \{x, y, z\}}\xi^i_j\sigma_j$, we can represent it as
\begin{equation*}
    r_i = \displaystyle\sum_{j \in \{x, y, z\}}\xi^i_j r_{j0}.
\end{equation*}
Then matrix $\|\xi^i_j\|$ can be easily found by calculating~$U^\dagger\sigma_iU$ using algebraic properties of Pauli matrices:
\begin{equation*}
   \|\xi^i_j\| = \begin{pmatrix}
       \dfrac{a^2 + \cos{2\alpha}}{1 + a^2} & -\dfrac{\sin{2\alpha}}{\sqrt{1 + a^2}} & \dfrac{a}{1 + a^2}(1 - \cos{2\alpha}) \\
       \dfrac{\sin{2\alpha}}{\sqrt{1 + a^2}} &
       \cos{2\alpha} &
       -\dfrac{a}{\sqrt{1 + a^2}} \sin{2\alpha}\\
       \dfrac{a}{1 + a^2}(1 - \cos{2\alpha}) &
       \dfrac{a}{\sqrt{1 + a^2}} \sin{2\alpha} &
       \dfrac{1 + a^2\cos{2\alpha}}{1 + a^2}
   \end{pmatrix},
\end{equation*}
which corresponds to matrix exponential~(\ref{matrix_exponential_gamma=0}) and hence the result of the present work agrees with the coherenc control case of closed two-level quantum systems.

\subsection{Exact analytical formula for gradient-based optimization for a qubit}
We summarize the above results as the following theorem.

\begin{theorem}
    If the system~(\ref{system_qubit}) is controlled by piecewise control~(\ref{PC qubit control: u}) and (\ref{PC qubit control: w}), then the derivatives of the final Bloch vector $\mathbf{r}_T = \r^M$ with respect to the components of control~(u, w) are
    \begin{equation} 
        \pd{\mathbf{r}_T}{u_j}  = S_N e^{\Lambda_N \omega\Delta t_N} S_N^{-1}\dots S_{j + 1}e^{\Lambda_{j + 1} \Delta t_{j + 1}}S_{j + 1}^{-1}\Bigg[\, \pd{}{u_j}\left(e^{A_j \Delta t_j}\right)\,\r_{j-1} + \pd {\mathbf{g}_j}{u_j}\Bigg],
        \label{final state gradient u diagonalized}
    \end{equation}
    \begin{equation}
        \pd{\mathbf{r}_T}{w_j}  = S_N e^{\Lambda_N \omega\Delta t_N} S_N^{-1}\dots S_{j + 1}e^{\Lambda_{j + 1} \Delta t_{j + 1}}S_{j + 1}^{-1}\Bigg[\, \pd{}{w_j}\left(e^{A_j \Delta t_j}\right)\,\r_{j-1} + \pd {\mathbf{g}_j}{w_j}\Bigg].
        \label{final state gradient w diagonalized}
    \end{equation}
Here matrices $S_k$ and $e^{\Lambda_k \omega\Delta t_k}$ are determined by control~$(u_k, w_k)$ and have one of the three possible forms~(\ref{case Delta neq 0}),~ (\ref{case Delta = 0}),~and~(\ref{case Delta p = q = 0}) which correspond to the cases 1, 2, 3, depending on the values of the coefficients $\Delta$, $p$, and $q$ for the control~$(u_k,w_k)$. The derivatives in the formulae~(\ref{final state gradient u diagonalized})~and~(\ref{final state gradient w diagonalized}) are
    \begin{align}
        \pd {\mathbf{g}_k}{u_k} &= \left(\pd{}{u_k}e^{A_k \Delta t_k} - (S_k e^{\Lambda_k \omega\Delta t_k} S_k^{-1} - \I)A_k^{-1} B^u\right) A_k^{-1} b, 
        \nonumber\\ 
        \pd{}{u_k}e^{A_k \Delta t_k} &= \Delta t_k\int \limits _0^1 S_k \mathrm{exp}({\alpha\Lambda_k \omega\Delta t_k}) S_k^{-1}\, B^u\, S_k \mathrm{exp}({(1-\alpha)\Lambda_k \omega\Delta t_k} S_k^{-1})\d\alpha, \label{gradient_exp_u_diag}\\ 
        \pd {\mathbf{g}_k}{w_k} &= \left(\pd{}{w_k}e^{A_k \Delta t_k} - 2 w_k (S_k e^{\Lambda_k \omega\Delta t_k} S_k^{-1} - \I)A_k^{-1} B^n\right) A_k^{-1} b, \nonumber\\ 
        \pd{}{w_k}e^{A_k \Delta t_k} &= 2 w_k \Delta t_k\int \limits _0^1 S_k \mathrm{exp}({\alpha A_k \Delta t_k })S_k^{-1}\, B^n\, S_k\mathrm{exp}((1-\alpha)\Lambda_k \omega\Delta t_k) S_k^{-1}\d\alpha. \label{gradient_exp_w_diag}
    \end{align}
\end{theorem} 

Note that while the formulae in the theorem complex, they are analytical and do not require solution of the differential equation which determines evolution of the system. They are based on expilcit exact analytical expressions for matrix eigenvalues and eigenvectors and thus they give an exact analytical solution for any gradient-based optimization for an open two-level quantum system driven by coherent and incoherent controls.

\begin{remark}
    Derivatives of the matrix exponentials~(\ref{gradient_exp_u_diag})~and~(\ref{gradient_exp_w_diag}) can be calculated by differentiating~(\ref{matrix_exponential_diagonalized}). However, it can be done only at points where $\Delta\ne 0$ and the roots are differentiable. 
\end{remark}

\section{GRAPE algorithm for coherent and incoherent control: Numerical simulations}\label{Sec7}

\subsection{Linear in time gradient computation}
$O(M)$ time. For example, consider formula~(\ref{final state gradient u}) can be rewritten as follows:
$$\pd{\r_T}{u_j}  = R_j \left( \pd{}{u_j}\left (e^{A_j \Delta t_j} \right) \r^{j-1} + \pd {\mathbf{g}_j}{u_j}\right),$$
where we introduced the following sequence of matrices $R_j$:
\begin{align*}
    R_j &= e^{A_M \Delta t_M} \dots e^{A_{j + 1} \Delta t_{j + 1}} = \displaystyle \prod _{k = j + 1} ^M e^{A_k \Delta t_k},\quad j = 0,1,\dotsc,M\\
    R_M &= \I, \\
    R_{M-1} &= e^{A_M \Delta t_M}, \\
    &\;\;\vdots \\
    R_0 &= e^{A_M \Delta t_M} \dotsc e^{A_1 \Delta t_1}.
\end{align*}
Recurrent formula which allows to compute every element in M steps:
$$
R_{j-1} = R_j e^{A_j \Delta t_j},\quad j = 1,\dotsc,M.
$$
Intermediate states $\r^j$ also can be computed using recurrent equation (\ref{recurrent eq r}). Eventually,  derivatives $\smallpd{e^{A_j \Delta t_j}}{[u_j,w_j]}$ and $\smallpd {\mathbf{g}_j}{[u_j,w_j]}$ are given by (\ref{g gradient u}),~(\ref{F gradient u}),~(\ref{g gradient w}),~(\ref{F gradient w}). 
Thus, derivatives with respect to coherent and incoherent control can be computed in linear time.

\subsection{Results of numerical simulation}
In order to test the obtained expressions for gradient computation we consider the following numerical optimization problem. Let $\rho_{\textrm{target}}$ be a particular target state of the system (\ref{system_qubit}) and $\rho_T^{u,n}$ be final state of the system driven by controls $u$ and $n$. Consider the problem of finding controls $u$ and $n$ that steer the system to $\rho_\textrm{target}$, i.e. such that
$$\rho_T^{u,n} = \rho_\textrm{target}.$$

This problem can be formulated as minimization of the objective functional~(\ref{functional_distance}):
\begin{equation}
    \F[u,n] = \J_2(\rho_T^{u,n}) = \|\rho_T^{u,n} - \rho_\textrm{target}\|^2 \to \min_{u,n}.   
    \label{optimization_problem}
\end{equation}  
We consider the norm induced by Hilbert-Schmidt scalar product, in this case it coincides with the squared norm of vector $\frac{1}{\sqrt{2}}(\mathbf{r}_T^{u,n} - \mathbf{r}_\textrm{target})$:
\begin{equation}
    \|\rho_T^{u,n} - \rho_\textrm{target}\|^2 = \mathrm{Tr}(\rho_T^{u,n} - \rho_\textrm{target})^2 =\frac{1}{2}|\mathbf{r}_T^{u,n} - \mathbf{r}_\textrm{target}|^2 \to \min_{u,n}.
    \label{optimization_problem_bloch}
\end{equation}
If and only if the system can be steered from the state~$\mathbf{r}_0$ to the state~$\mathbf{r}_\textrm{target}$ under considered system parameters, the minimum should equal zero: $J(\rho_T^{u,n}) = 0$, i.e. $\mathbf{r}^{u,n}_T = \mathbf{r}_\textrm{target}.$

In order to solve the considered optimization problem numerically, we implement ordinary gradient descent method. This is iterative algorithm for finding minimum (or maximum) of the functional. Generally, this method makes discrete steps along gradient of the functional, therefore it needs computation of the gradient in each iteration. In some cases one can compute gradient approximately by finding finite differences for each component, however our obtained expressions for gradient allows us to compute exact value of gradient in any point of control space. In the context of considered optimization problem~(\ref{optimization_problem_bloch}) gradient~(\ref{functional_gradient_bloch}) with respect to control~$u$~(\ref{PC qubit control: u}) and $w$~($n = w^2$)~(\ref{PC qubit control: w}) equals
\begin{equation}
    \dfrac{\partial \F}{\partial [u_j, w_j]} = (\mathbf{r}_T^{u,w} - \mathbf{r}_\textrm{target}) \pd{\mathbf{r}_T^{u,w}}{[u_j, w_j]}, \quad j = 1,\dots,M;
    \label{functional_gradient_bloch_transition}
\end{equation}
where derivatives w.r.t. each of $2M$~components of control $\pd{\mathbf{r}_T^{u,w}}{[u_j, w_j]}$ are given by expressions~(\ref{final state gradient u})~and~(\ref{final state gradient w}) or their diagonalized versions~(\ref{final state gradient u diagonalized})~and~(\ref{final state gradient w diagonalized}). For numerical simulations we used expressions without diagonalization, i.e. corresponding matrix exponentials were found numerically.

The gradient descent has the following iterative formula:
\begin{equation}
	(u^{(k + 1)}, w^{(k + 1)}) = (u^{(k)}, w^{(k)}) - h_k \mathrm{grad}_{u,w} \F(u^{(k)}, {w^{(k)}}^2), \qquad k = 0, 1, \dots;
	\label{GRAPE_gradient_descent}
\end{equation}
where $u^{(k)}$ and $w^{(k)}$ are control on $k$-th iteration, $h_k$ are the values of iterations steps,  $\mathrm{grad}_{u,v} \F(u, w)$ has $2M$ components which should be calculated by~(\ref{functional_gradient_bloch_transition}). Iterations of gradient descent method continues until the following stop criterion is satisfied:
\begin{equation}
    \F^{(k)} = \F[u^{(k)}, n^{(k)}] = \|\rho_T^{u,n} - \rho_\textrm{target}\|^2 < \varepsilon_1,
    \label{stop_criterion_functional}
\end{equation}
where $\varepsilon_1 > 0$ --- small parameter of accuracy, thus we find control that steers the system to the state~$\rho_T$ which differ from~$\rho_\textrm{target}$ by less than~$\varepsilon_1$ in the sense of Hilbert-Schmidt distance. Moreover, for cases when the target state is unreachable we use standart stop criterion
\begin{equation}
    \left|\mathrm{grad}_{u,w} \F\left[u^{(k)}, {w^{(k)}}^2\right]\right| < \varepsilon_2,
    \label{stop_criterion_gradient}
\end{equation}
that corresponds to the position of local extremum, i.e. where gradient equals zero (first order necessary optimality condition).
\begin{figure}
    \centering
    \includegraphics[width = \linewidth]{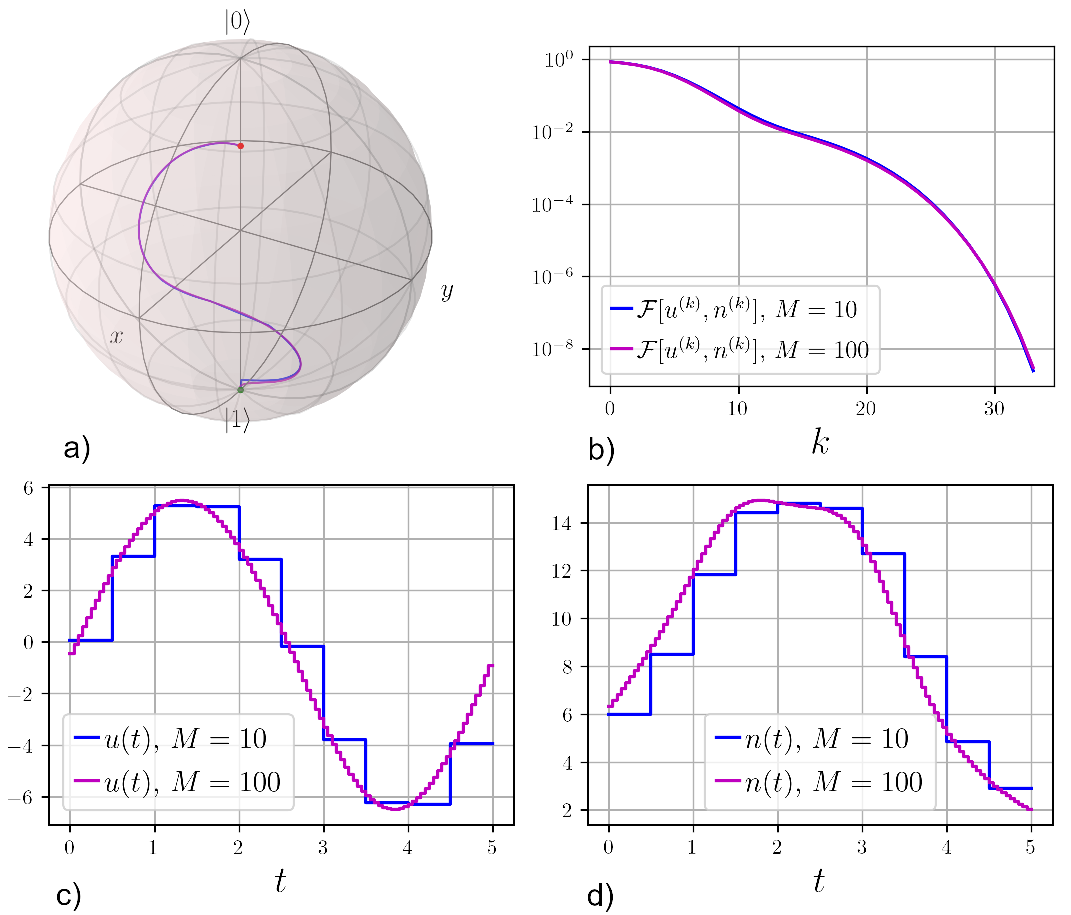}
    \caption{The initial and the target states $\rho_0={\rm diag}(0,1)$ and $\rho_{\rm tagret}={\rm diag}(3/4,1/4)$, shown on the subplot (a) by green and red dots, respectively. This target state is reachable from the initial state according to~\cite{Lokutsievskiy_2021}. Blue and purple lines show optimization results for $M=10$ and $M=100$, respectively. Subplot (a) shows optimal trajectories in the Bloch ball; they almost coincide. Subplot (b) shows fidelity vs number of iterations; the convergence is fast, infidelities $2.5\times10^{-9}$ for $M = 10$ and $3.1 \times 10^{-9}$ for $M = 100$ are obtained after 34 iterations for both cases. Subplots (c) and (d) show the obtained after optimization coherent and incoherent controls, respectively. Both these controls are non-zero.}
    \label{FigOptimization1}
\end{figure}

\begin{figure}
    \centering
    \includegraphics[width = \linewidth]{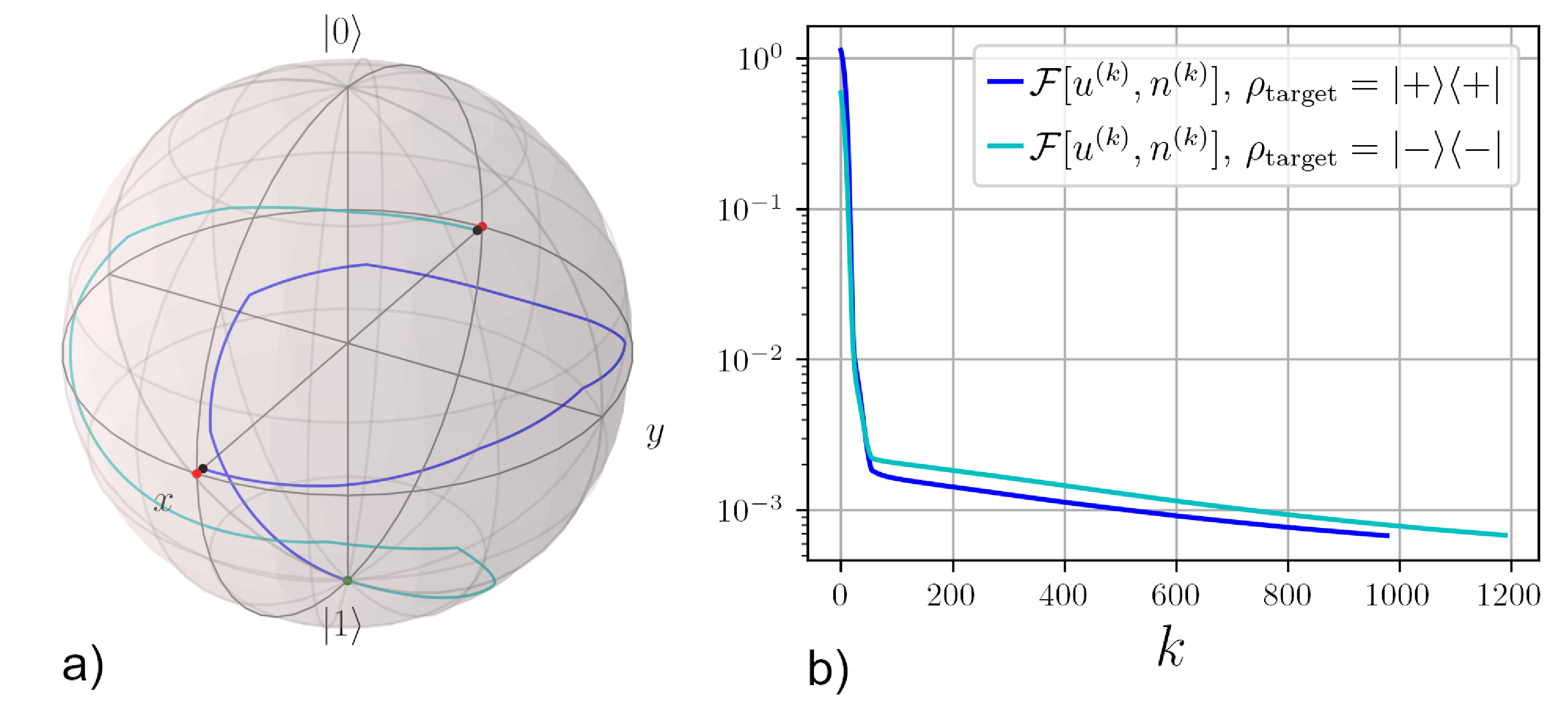}
    \caption{Two trajectories with target states in the two most unreachable points shown by red dots.
    The initial state $\rho_0=|0\rangle\langle 0|$ is shown on the subplot (a) by green dot. The two target states $\rho_{\rm tagret}=|\pm\rangle\langle\pm|$, where $|\pm\rangle=(|0\rangle\pm|1\rangle)/\sqrt{2}$, are shown on the subplot (a) by two red dots. These target states are unreachable from the initial state according to the result of~\cite{Lokutsievskiy_2021}. Numerical results agree with this finding, as the optimization algorithm for these target states does not converge to low fidelity values. Black dots show closest points to the target states obtained after the optimization. Minimal infidelity value~$\F = 6.7\times10^{-4}$ is obtained after $978$ iterations for $|+\rangle$ target state and $\F = 6.8\times10^{-4}$ after $1190$ iterations for $|-\rangle$ target state.}
    \label{FigOptimization2}
\end{figure}

Numerical simulation are made with the parameters relative scale of which is chosen in accordance with physical situations, so that transition frequency $\omega = 1$, strength of interaction with the control field $\mu = 0.1$, strength of interaction with the environment $\gamma = 0.01$. We used regular partition of time segment~$[0,T]$ into $M$~intervals, we considered time of evolution $T = 5$ and number of intervals $M = 10$ and $M = 100$. 

As initial guess we chose coherent control  $u^{(0)}_j = \sin{(2\pi t_{j-1} / T)}$ and incoherent control $w^{(0)}_j = \exp{(-4 (t_{j-1}/T - 1/2)^2)}$, $j = 1, \dots, M$, which are similar to the form of physical representation of coherent and incoherent control. For stop criterion small parameters~(\ref{stop_criterion_functional})~and~(\ref{stop_criterion_gradient}) we set accuracy~$\varepsilon_1 = 10^{-4}$ (first criterion) and $\varepsilon_2 = 5 \cdot 10^{-3}$ (second criterion).
We consider adaptive iteration steps~$h_k$ when value of step depends on if the value of the functional decrease or not, so that we get strictly decreasing sequence of the functional values. Literally, we take some initial step $h_0$ and for each iteration if value of the functional decreases: $\F^{(k + 1)} < \F^{(k)}$, than move on to next iteration and increase value of the step: $h_{k+1} = \alpha h_k$, $\alpha \geq 1$; else value of the functional increases, so we repeat iteration with reduced step: $h_{k}:= \beta h_{k}$, $0 < \beta < 1$ until value of step is small enough to successful decrease of the functional value. We chose the values of parameters as $h_0 = 10,\,100$ for example on Fig.\ref{FigOptimization1} and $h_0 = 1$ for example on Fig.\ref{FigOptimization2}, $\alpha = 1.1$ and~$\beta = 0.5$. We present several examples with different initial and target states.

First, as initial and target states we choose diagonal states
\[
\rho_0=|1\rangle\langle 1|=
\begin{pmatrix}
0 & 0\\ 0 & 1
\end{pmatrix},\qquad 
\rho_{\rm tagret}=
\begin{pmatrix}
3/4 & 0\\ 0 & 1/4
\end{pmatrix},
\]
corresponding Bloch vectors are $\r_0 = (0, 0, -1)$ and $\r_\textrm{target} = (0, 0, 0.5)$. 
This target state is reachable from the initial state, as was proven in~\cite{Lokutsievskiy_2021}. The numerical optimization confirms this result. The corresponding points in the Bloch ball are shown on Fig.~\ref{FigOptimization1} by green and red markers, respectively. Optimal evolutions for $M=10$ and $M=100$ segments of the total time interval $[0,T]$ are shown on the subplot (a) by blue and purple colours, respectively. The optimal trajectories almost coincide. The corresponding fidelity values vs the number of iteration steps are plotted on subplot (b). Coherent and incoherent controls found as the result of the optimization are plotted on subplots (c) and (d), correspondingly. The convergence is fast, infidelities $2.5\times10^{-9}$ for $M = 10$ and $3.1 \times 10^{-9}$ for $M = 100$ are obtained after 34 iterations for both cases.

In the second example, the initial state is the same  $\rho_0=|0\rangle\langle 0|$, but the target states are $\rho_{\rm tagret}=|\pm\rangle\langle\pm|$, where $|\pm\rangle=(|0\rangle\pm|1\rangle)/\sqrt{2}$. According to~\cite{Lokutsievskiy_2021}, these target states are unreachable from the initial state. The optimization results are shown on figure~\ref{FigOptimization2}. Numerical results agree with this finding, since the optimization algorithm for these target states does not converge to low infidelity values. The Bloch vectors of the target states are~$\r_{\rm target} = (\pm1, 0, 0)$ and located on $x$-axis. The initial and two target states are shown on the subplot (a) by green dot and two red dots, respectively. Black dots show closest points to the target states obtained after the optimization. The algorithm is not able to find optimal controls; minimal infidelity value~$\F = 6.7\times 10^{-4}$ is obtained after $978$ iterations for the $|+\rangle$ target state and $\F = 6.8\times 10^{-4}$ after $1190$ iterations for the $|-\rangle$ target state. This infedility values mean that the distance between the final state $\r_T^f$ and the target state $\r_{\rm target}$ is of order $10^{-2}$ that agrees with the theoretical result in~\cite{Lokutsievskiy_2021}. According to this result the minimal value of the distance between the final state and the target state is of order $\gamma/\omega = 10^{-2}$.

One can see that on figure~\ref{FigOptimization2}b there is an inflection point. At the iteration $k = 55$ the algorithm attains the functional value $\F = 1.8\times10^{-3}$ for the $|+\rangle$ target state and $\F = 2.2\times10^{-3}$ for the $|-\rangle$ target state and then convergence speed sharply decreases. This can be explained by the fact that the step value $h_k$ firstly increases with iterations until it reaches critical value when the corresponding gradient descent iteration does not reduce the functional value, and after the step value does not increase but fluctuates. This step values behavior affects speed of convergence.

All computation were performed via NumPy Python library which enables fast matrix computations, matrix exponentials, in cases when they were not diagonalized, were computed numerically using function \texttt{scipy.linalg.expm} of SciPy library. The integrals~(\ref{F gradient u})~and~(\ref{F gradient w}) are computed numerically by standart trapezoidal rule.

\section{Conclusions}\label{Sec8}
The GRAPE method is widely used for optimization in quantum control and has been applied to coherently controlled closed and open quantum systems. In this work, we adopt GRAPE method for optimizing objective functionals for open quantum  systems driven by both coherent and incoherent controls; here incoherent control represents spectral density of the tailored environment which acts on the system as control. Incoherent control by its physical meaning is non-negative and hence is bounded from below. For this problem we compute gradient of the objective for a general case of $N$-level open system and for various objectives. The gradient is computed for piecewise constant class of control. The case of a single qubit is considered in details and solved analytically. For this case, an explicit analytical expression for evolution and gradient is obtained via diagonalization of the $3\times 3$ matrix determining the controlled system dynamics in the Bloch representation. The diagonalization is obtained by solving a cubic equation via the Cardano's method. The efficiency of the algorithm is demonstrated through numerical simulations for the state-to-state transition problem and its complexity is estimated. The obtained analytical expression for the gradient can be used in any optimization methods which are based on gradient, including GRAPE, BFGS, etc.

\vspace{6pt} 

\section*{Acknowledgements}
We thank Boris O. Volkov for  a useful discussion of the robustness issue. This work was funded by the Ministry of Science and Higher Education project \textnumero~075-15-2020-788.

\end{document}